\documentclass[sigconf,natbib=true]{acmart}

\AtBeginDocument{
  }

\usepackage{multirow}
\usepackage{booktabs}
\usepackage{tabularx}
\usepackage{xcolor}
\usepackage{array}
\usepackage{tikz}
\usepackage{url}

\newcommand{\myarrow}{%
    \tikz[baseline=-0.5ex]{ 
        \draw[->, >=stealth, line width=0.6pt] (0,1.2ex) -- (0,0) -- (1em,0);
    }%
}

\copyrightyear{2026}
\acmYear{2026}
\setcopyright{cc}
\setcctype{by}
\acmConference[SIGIR '26]{Proceedings of the 49th International ACM SIGIR Conference on Research and Development in Information Retrieval}{July 20--24, 2026}{Melbourne, VIC, Australia}
\acmBooktitle{Proceedings of the 49th International ACM SIGIR Conference on Research and Development in Information Retrieval (SIGIR '26), July 20--24, 2026, Melbourne, VIC, Australia}
\acmDOI{10.1145/3805712.3809724}
\acmISBN{979-8-4007-2599-9/2026/07}

\begin{document}

\title{Why Advanced Encoders Lag on Sparse Retrieval? The Answer and an Approach to Bridging Vocabulary Gaps}

\author{Zhichao Geng}
\affiliation{
  \institution{Amazon Web Service}
  \city{Shanghai}
  \country{China}
}
\email{zhichaog@amazon.com}

\author{Yang Yang}
\affiliation{%
  \institution{Amazon Web Service}
  \city{Shanghai}
  \country{China}
}
\email{yych@amazon.com}

\renewcommand{\shortauthors}{Geng et al.}

\begin{abstract}
While advanced foundation models like ModernBERT significantly outperform older architectures in dense retrieval, they surprisingly lag behind the aging BERT-base baseline in learned sparse retrieval (LSR). We identify the root cause as the \textit{Vocabulary Gap}: modern tokenizers utilize raw, case-sensitive vocabularies designed for lossless reconstruction, which map single semantic units to redundant surface forms, wasting model capacity on morphological noise and hindering lexical matching. We formalize this intuition through a theoretical framework, demonstrating that appropriate vocabulary coarse-graining can tighten the generalization bounds by reducing complexity of the hypothesis class, provided that semantic integrity is preserved. 
To resolve this, we propose \textbf{Vocabulary Transfer (VT)}, a model-agnostic framework that migrates advanced encoders to sparse-friendly, normalized vocabularies with minimal computational cost.
 VT utilizes a novel \textbf{Semantic Initialization} via spatial topology to preserve geometric structure and an \textbf{Activation Potential Calibration (APC)} mechanism to align pre-trained manifolds with sparsity constraints, preventing the dead neuron and dense collapse observed in standard fine-tuning. Empirically, VT is universally effective: it enables ModernBERT to achieve state-of-the-art performance on the BEIR benchmark (\textbf{52.4} nDCG, a \textbf{+4.7} improvement), resuscitates failing models like RoBERTa-large, and generalizes seamlessly to inference-free architectures and specialized domains. These results confirm that the performance lag is not an architectural deficiency but a solvable vocabulary mismatch. We've released our code and models.\footnote{\url{https://anonymous.4open.science/r/vocab-transfer/}. All details included.}
\end{abstract}



\maketitle

\section{Introduction}

The landscape of neural information retrieval has bifurcated into two dominant paradigms: dense retrieval, which encodes queries and documents into continuous low-dimensional embeddings \cite{dpr, ance}, and learned sparse retrieval (LSR), which projects text into high-dimensional, weighted lexical vectors \cite{formal2021splade, deepimpact}.
While dense retrievers excel at capturing semantic nuances, sparse retrievers—exemplified by models like SPLADE \cite{formal2021splade}—retain the interpretability and efficiency of inverted indices while mitigating the lexical mismatch problem of traditional BM25~\cite{robertson1995okapi, manning2008ir}.

In the dense retrieval paradigm, upgrading the backbone is a proven strategy.
Modern foundations like ModernBERT \cite{warner2025smarter} provide not only stronger representations but also architectural advantages like 8k context windows and FlashAttention compatibility.

\begin{figure}[t]
    \centering
    \includegraphics[width=\linewidth]{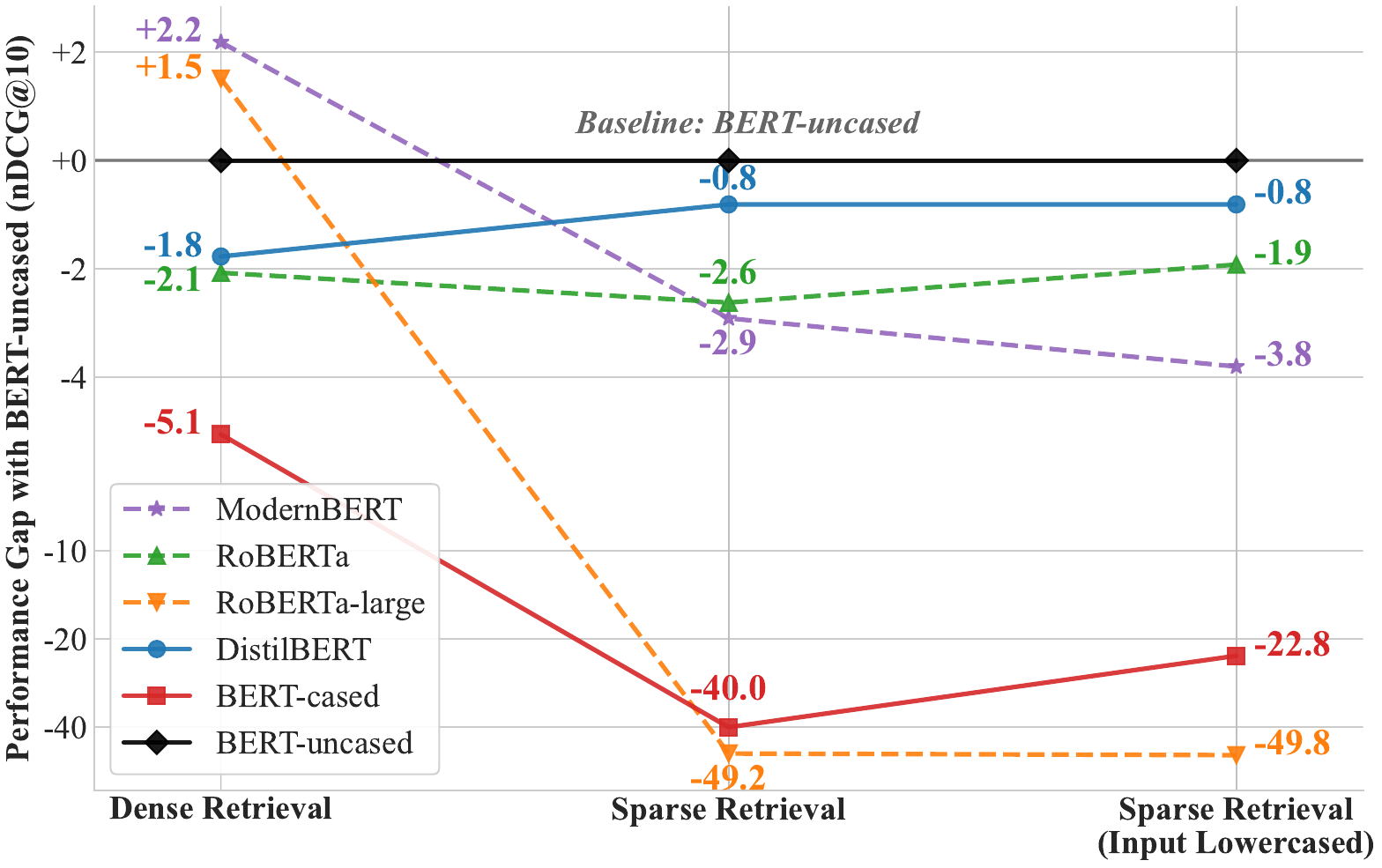} 
    \caption{The \textbf{Vocabulary Gap} anomaly.
    While advanced encoders like ModernBERT significantly outperform BERT in dense retrieval, they lag behind in sparse retrieval under standard fine-tuning.}
    \label{fig:motivation}
\end{figure}

However, these architectural leaps remain inaccessible to sparse retrieval.
We observe a puzzling anomaly: \textit{advanced encoders consistently underperform in sparse settings, often lagging behind the older BERT-base-uncased baseline.}
As illustrated in \autoref{fig:motivation}, this performance degradation is pervasive.
The most intuitive explanation attributes this to the BPE tokenizer differences in modern models.
However, we observe that \path{bert-base-cased}, which uses the same WordPiece tokenizer as the effective \path{bert-base-uncased} baseline, performs equally poorly.
This isolates the degree of vocabulary normalization as the critical variable.
This regression persists despite identical training pipelines, suggesting that the architectural advancements of modern backbones are stifled by a fundamental incompatibility with the sparse retrieval objective.

We identify the root cause as the \textbf{vocabulary gap}—specifically, the shift in modern tokenization toward raw vocabularies (i.e., lacking normalization or pre-tokenization) designed for lossless reconstruction.
These tokenizers map single semantic units to redundant surface variants (e.g., ``Token'' vs. ``token''), forcing the model to waste capacity bridging these orthogonal dimensions—a burden dense models bypass.
While forcing input lowercasing offers partial relief, it is insufficient;
aggressive lowercasing on a case-sensitive tokenizer often fragments tokens (e.g., \texttt{Halloween} $\rightarrow$ \texttt{hall}, \texttt{ow}, \texttt{een}), destroying semantic integrity.

Compounding this challenge is the prohibitive cost of remediation.
While training a model from scratch with a sparse-friendly vocabulary could theoretically solve the issue, it is computationally impractical.
Modern foundation models are trained on massive corpora—ModernBERT, for instance, on 2 trillion tokens~\cite{warner2025smarter}.
Replicating this pre-training scale simply to swap the vocabulary is infeasible for most applications.
Consequently, the field faces a dilemma: we require the reasoning power and inference efficiency of modern backbones, yet their native vocabularies are ill-suited for sparse retrieval.

In this work, we provide the answer to this lag and a method to resolve it.
We argue that sparse retrieval requires a \textit{Representation-Compatible} vocabulary—one that normalizes surface forms while preserving semantic distinctions.
We formalize this intuition through a theoretical framework showing that appropriate vocabulary coarse-graining improves the generalization bound of sparse retrievers by reducing hypothesis class complexity without sacrificing approximation power.

Guided by this theory, we propose \textbf{Vocabulary Transfer (VT)}, a recipe to migrate strong pre-trained backbones to a sparse-friendly vocabulary with \textbf{minimal cost}---using $<0.2\%$ of the original ModernBERT training tokens and achieving near-optimal performance with just 500 MLM steps.
VT utilizes a novel \textbf{Semantic Initialization} via spatial topology and an \textbf{Activation Potential Calibration} mechanism. 
This aligns the advanced backbone with the sparsity constraints of models like SPLADE, preventing the ``dead neuron'' and dense collapse observed in standard fine-tuning.

Our contributions are as follows:
\begin{itemize}
    \item \textbf{Theoretical Analysis:} We derive a generalization bound for sparse retrieval under vocabulary coarse-graining, introducing \textit{Representation Compatibility} (RC) to explain why normalization improves learnability.
    \item \textbf{Methodology:} We propose VT, a model-agnostic procedure that transplants regularized vocabularies onto advanced encoders using geometric initialization and discrepancy-aware adaptation.
    \item \textbf{Empirical Validation:} We demonstrate that VT is \textbf{universally effective}. 
    It enables ModernBERT to achieve state-of-the-art results on BEIR~\cite{thakur2021beir} (52.4 nDCG, a \textbf{+4.7} improvement), \textbf{resuscitates} failing models like RoBERTa-large, and generalizes seamlessly to \textbf{inference-free} architectures and \textbf{domain-specific} adaptation.
\end{itemize}

\section{Related Work}

\subsection{Neural Sparse Retrieval}
The evolution of information retrieval has seen a transition from exact matching heuristics, such as BM25 \cite{robertson2009probabilistic}, to neural architectures that learn semantic representations.
While dense retrieval \cite{karpukhin2020dense, xiong2020approximate} encodes queries and documents into continuous low-dimensional spaces, Learned Sparse Retrieval (LSR) projects text into high-dimensional sparse vectors, preserving the interpretability and efficiency of inverted indices.

Early LSR approaches focused on estimating term weights or expanding documents with relevant terms.
\textbf{DeepCT} \cite{dai2020context} utilized BERT to predict context-aware term weights, mapping them back to the bag-of-words space.
Similarly, \textbf{docT5query} \cite{nogueira2019document} employed generative models to expand documents with potential queries.
\textbf{SparTerm} \cite{bai2020sparterm} introduced a gating mechanism to explicitly learn term importance and enforce sparsity.
\textbf{COIL} \cite{gao2021coil} bridged the gap between sparse and dense methods by storing efficient contextualized representations in inverted lists.

The SPLADE family \cite{formal2021splade, splade2, lassance2022efficiency, lassance2024splade} represented a paradigm shift by applying sparsity regularization directly on the Masked Language Model (MLM) logits, performing simultaneous expansion and weighting.
Recent research has shifted towards \textit{inference-free} architectures to reduce query-side latency.
TILDE \cite{zhuang2021tilde} and subsequent works \cite{geng2024towards, shen2025exploring} pre-compute document representations while keeping query processing lightweight.
However, these models are exposed to the \textbf{Vocabulary Gap}, as they lack the capacity to dynamically bridge lexical mismatches between the pre-trained backbone and the retrieval task.

\subsection{Pre-trained Backbones and Tokenization}
The efficacy of Pre-trained Language Models (PLMs) is inextricably linked to their tokenization strategies.
Standard architectures like BERT \cite{devlin2019bert} utilize WordPiece \cite{schuster2012japanese}, while modern backbones such as RoBERTa \cite{liu2019roberta} and ModernBERT \cite{warner2025smarter} rely on BPE~\cite{sennrich2016neural}.
While techniques like Subword Regularization \cite{kudo2018subword} and CharacterBERT \cite{el2020characterbert} attempt to improve morphological robustness, the rigid distinctness of surface forms in standard subword vocabularies remains a fundamental bottleneck for sparse matching, necessitating significant model capacity to bridge these lexical gaps.

\subsection{Vocabulary Transfer and Adaptation}
Adapting pre-trained models to new vocabularies is a critical challenge.
This problem has been extensively studied in the context of cross-lingual transfer, where vocabulary misalignment severely hampers performance \cite{artetxe2020cross}.
To address this, various initialization strategies have been proposed to align new vocabularies with pre-trained manifolds without full retraining.
WECHSEL \cite{minixhofer2022wechsel} used a shared bilingual static embedding space to map target subwords and initializes each new token embedding as a similarity-weighted average of its $k$ nearest source subword embeddings.
More recently, FOCUS \cite{dobler-de-melo-2023-focus} was brought up to transfer the vocabulary from monolingual language model to multilingual. It leveraged FastText~\cite{bojanowski2017enriching} to derive similarity relations between the new token and anchor tokens, which are then used to weight the combination.
Mundra et al. \cite{mundra2024empirical} provided a comprehensive empirical validation of these strategies, highlighting that leveraging the source embedding structure is crucial for convergence.

In the specific context of Learned Sparse Retrieval (LSR), the impact of vocabulary design is profound yet only recently gaining attention. \citet{lionis2026case} empirically confirmed the effect of vocabulary casing on sparse retrieval, and \citet{lei2025enhancing} explored enhancing lexicon-based embeddings with LLMs.
Regarding adaptation, ESPLADE \cite{dudek2023learning,kim2025role} represents a recent attempt to transfer SPLADE capabilities to new vocabularies.
However, ESPLADE relies on computationally expensive continuous Masked Language Modeling (MLM) on large corpus to align the new embedding space.
Unlike these approaches, our work proposes a Representation-Compatible (RC) transfer method that utilizes geometric initialization to close the vocabulary gap with minimal adaptation cost.

\subsection{Theoretical Analysis of Retrieval Models}

Theoretical analyses of retrieval models traditionally focus on probabilistic relevance modeling and term-weighting schemes such as BM25 within the probabilistic relevance framework~\cite{robertson2009probabilistic, manning2008ir}.
For modern neural models, most available theory comes from general learning-theoretic tools rather than IR-specific analyses.
In particular, Rademacher-complexity-based bounds for linear predictors with $\ell_1$ and $\ell_2$ constraints provide sharp estimates of sample complexity and margin-based generalization for sparse linear hypothesis classes~\cite{Kakade2009}.
These results underpin many later analyses of regularization, sparsity, and high-dimensional learning.
For learned sparse retrieval, existing work has focused largely on empirical or architectural aspects.
To the best of our knowledge, there is still little work that explicitly connects vocabulary design and normalization to capacity measures or sample complexity in LSR.

\section{Theoretical Analysis}
\label{sec:theory}

We give a unified analysis for sparse retrievers that operate in a shared discrete keyspace (tokens/terms) with nonnegative, sparse weights on both sides.
To avoid overloading symbols, we reserve $d$ for documents and use $p$ for feature-space dimensionality throughout this section.

\subsection{Modeling via RC Coarse-Graining}
\label{subsec:model-rc}

While neural sparse retrievers like SPLADE utilize dot-products between two learned encoders, we analyze the generalization capability of the document encoder by treating the query encoder as generating a distribution of linear weights. This standard reduction allows us to apply Rademacher complexity analysis to the sparse representation learning problem.

\paragraph{Sparse keyspace.}
Let $V$ be a discrete keyspace.
For query $q$ and document $d$, let $w_{\theta,q},w_{\theta,d}\in\mathbb{R}_{\ge0}^{V}$ be sparse encoder weights.
We consider separable per-key features
\[
[u_\theta(q,d)]_t \triangleq \psi(w_{\theta,q}(t))\,\phi(w_{\theta,d}(t)),
\qquad u_\theta(q,d)\in[0,R]^{|V|},
\]
where $\psi,\phi:\mathbb{R}_{\ge0}\to\mathbb{R}_{\ge0}$ are non-decreasing.
Assume $\ell_1$ budgets (encouraged by sparsity regularization)
$\|w_{\theta,q}\|_1\le S_q$, $\|w_{\theta,d}\|_1\le S_d$,
which imply $\|u_\theta(q,d)\|_\infty\le R\triangleq \psi(S_q)\phi(S_d)$.

\paragraph{Coarse-graining.}
A normalizer induces a many-to-one map $\pi:V\to V'$.
Let $G\in\mathbb{R}_{\ge0}^{|V'|\times|V|}$ be row-stochastic and $\pi$-respecting:
$G_{ut}=0$ if $t\notin\pi^{-1}(u)$ and $\sum_{t\in\pi^{-1}(u)}G_{ut}=1$.
Define coarse-grained features
\[
u'_\theta(q,d)\triangleq G\,u_\theta(q,d)\in[0,R]^{|V'|}.
\]

\paragraph{Hypothesis classes.}
With a shared $\ell_1$ budget $B$, define
\begin{align}
\mathcal{H}_V
&=\{\langle\beta,u_\theta(q,d)\rangle:\ \beta\!\ge\!0,\ \|\beta\|_1\le B\}, \label{eq:HV}\\
\mathcal{H}_{V'}
&=\{\langle\beta',u'_\theta(q,d)\rangle:\ \beta'\!\ge\!0,\ \|\beta'\|_1\le B\}. \label{eq:HVp}
\end{align}

\paragraph{Representation-compatibility (RC)}
Overly aggressive coarse-graining (e.g., heavy stemming) can conflate meanings, so we focus on normalizers that mostly merge surface variants (e.g., case folding~\cite{manning2008ir}).
The aggregation $G$ is \emph{RC} if there exists $\varepsilon_{\mathrm{RC}}\ge0$ such that
for every $\beta\ge0$ with $\|\beta\|_1\le B$ there is $\beta'\ge0$,
$\|\beta'\|_1\le B$, satisfying
\begin{equation}
\sup_{(q,d)}
\big|\langle\beta,u_\theta(q,d)\rangle-\langle\beta',u'_\theta(q,d)\rangle\big|
\le \varepsilon_{\mathrm{RC}}.
\tag{RC$^\star$}
\label{eq:RC}
\end{equation}

\subsection{Generalization under RC Coarse-Graining}
\label{subsec:gen-rc}

\paragraph{Rademacher tools.}
Let $p\in\mathbb{N}$, $B,R>0$, and $\mathcal{G}_p=\{\langle\beta,x\rangle:\ \beta\!\ge\!0,\ \|\beta\|_1\le B\}$.
For a vector class $\mathcal{U}\subset[0,R]^p$, we use the $\ell_\infty$-type empirical Rademacher complexity
\[
\hat{\mathfrak{R}}_n(\mathcal{U};\|\cdot\|_\infty)
\;\triangleq\;
\mathbb{E}_{\sigma}\!\left[
  \sup_{u\in\mathcal{U}}
  \left\|
    \frac{1}{n}\sum_{i=1}^n \sigma_i\,u(z_i)
  \right\|_{\infty}
\right],
\]
where $z_{1:n}$ is the sample and $\sigma_{1:n}$ are i.i.d.\ Rademacher signs.
Then the following bounds hold~\cite{bartlett2002rademacher, Mohri2018}:
\begin{equation}
\label{eq:rad-composition}
\hat{\mathfrak{R}}_n(\mathcal{G}_p\!\circ\!\mathcal{U})
\ \le\ B\,\hat{\mathfrak{R}}_n(\mathcal{U};\|\cdot\|_\infty).
\end{equation}
If $G\in\mathbb{R}_{\ge0}^{p'\times p}$ is row-stochastic, then~\cite{horn2013matrix}
\begin{equation}
\label{eq:rad-contraction-G}
\hat{\mathfrak{R}}_n(G\!\circ\!\mathcal{U};\|\cdot\|_\infty)
\ \le\ \hat{\mathfrak{R}}_n(\mathcal{U};\|\cdot\|_\infty).
\end{equation}

\begin{lemma}[Row-stochastic aggregation does not increase feature-class complexity]
\label{lem:rad-rc}
Consequences for $u'_\theta=G\,u_\theta$.
Recall $\mathcal{W}_V=\{u_\theta(\cdot)\in[0,R]^{|V|}\}$ and $\mathcal{W}_{V'}=\{u'_\theta(\cdot)=G\,u_\theta(\cdot)\in[0,R]^{|V'|}\}$, and the linear heads $\mathcal{H}_V,\mathcal{H}_{V'}$ from \eqref{eq:HV}–\eqref{eq:HVp}.
By \eqref{eq:rad-contraction-G},
\begin{equation}
\label{eq:feature-class-no-increase}
\hat{\mathfrak{R}}_n(\mathcal{W}_{V'};\|\cdot\|_\infty)
\ \le\ \hat{\mathfrak{R}}_n(\mathcal{W}_V;\|\cdot\|_\infty),
\end{equation}
\end{lemma}

\begin{theorem}[Sample complexity improves under RC coarse-graining]
\label{thm:sc}
Let $\ell:\mathbb{R}\to[0,1]$ be $L$-Lipschitz and let $\hat h_G\in\mathcal{H}_{V'}$ be an ERM on $n$ samples.
Assume RC in \eqref{eq:RC}, which implies $\inf_{h\in\mathcal{H}_{V'}}\mathcal{L}(h)\ \le\ \inf_{h\in\mathcal{H}_V}\mathcal{L}(h)\ +\ L\,\varepsilon_{\mathrm{RC}}$.
Then for any $\delta\in(0,1)$, with probability at least $1-\delta$,
\begin{align}
\mathcal{L}(\hat h_G)-\inf_{h\in\mathcal{H}_V}\mathcal{L}(h)
\le\;& L\,\varepsilon_{\mathrm{RC}}
 + 4BL\,\hat{\mathfrak{R}}_n(\mathcal{W}_{V'};\|\cdot\|_\infty)\nonumber\\
& + 2C_{\mathrm{gen}}\sqrt{\tfrac{\log(1/\delta)}{n}}.
\label{eq:main-bound}
\end{align}
Moreover $|V'|\le |V|$ and \eqref{eq:feature-class-no-increase} hold, hence the feature-class complexity term does not increase while the estimation error bound tightens.
If, in addition, $\varepsilon_{\mathrm{RC}}$ is sufficiently small, then the overall generalization bound under coarse-graining is tighter.
\end{theorem}

The proof follows from standard Rademacher symmetrization, Ledoux--Talagrand contraction, and McDiarmid's inequality; full details are in Appendix~\ref{app:proof-thm-sc}.

\paragraph{Corollary (pointwise \& pairwise).}
The theorem holds verbatim for pointwise and pairwise training by replacing $u_\theta(q,d)$ with the triplet difference $u_\theta(q,d^+)-u_\theta(q,d^-)$.
Consequently, the $\ell_\infty$–type vector Rademacher complexity appearing in Theorem~\ref{thm:sc} does not increase by more than a factor of $2$.
Hence the generalization bound holds verbatim for both pointwise and pairwise training.

\subsection{Inference-Free as a Special Case}
\label{subsec:theory-doconly}
For inference-free sparse retrievers~\cite{shen2025exploring, geng2024towards, splade2}, the query-side weights are fixed incidence vectors determined by the query text (i.e., $w_{\theta,q}$ is replaced by $w_q$), so the preceding analysis applies verbatim.
Under RC coarse-graining $G$, the induced feature-class complexity under the $\ell_\infty$-type Rademacher measure does not increase, hence the estimation term in Theorem~\ref{thm:sc} is no larger.

\paragraph{\textbf{Takeaway.}}
This result suggests a trade-off: coarse-graining ($V \to V'$) prevents the inflation of the Rademacher complexity term inherent to large raw vocabularies, potentially tightening the bound if the approximation cost $\varepsilon_{\mathrm{RC}}$ (introduced by merging tokens) is kept minimal.

\section{Vocabulary Transfer (VT): From Theory to Minimal Migration Cost}
\label{sec:vt}

\subsection{Design Goals}
Advanced encoders such as ModernBERT often underperform in neural sparse retrieval due to vocabulary mismatch and excessive surface-form variability. 
Our theoretical analysis in Section~\ref{sec:theory} suggests that a representation-compatible, coarse-grained vocabulary yields better generalization guarantees. However, the theory describes the destination, not the path. Training such a model from scratch is costly. Therefore, our goal is to migrate a pretrained backbone from its source vocabulary $V$ to a more regularized, existing target vocabulary $V'$ (e.g., \path{bert-base-uncased}) at \textbf{minimal migration cost}. 
Our method adheres to two principles: 
(1) \textbf{Distributional Consistency}, ensuring initialized parameters preserve the statistical priors of the target domain; and 
(2) \textbf{Optimization Efficiency}, prioritizing the adaptation of new semantic units to pre-condition the model for downstream sparsity constraints.

\subsection{Method: A Three-Step VT Recipe}
Let the source model be $\mathcal{M}=(V,E, \mathbf{b})$, where $E \in \mathbb{R}^{|V| \times d}$ is the embedding matrix and $\mathbf{b} \in \mathbb{R}^{|V|}$ is the output bias vector representing unigram log-probabilities. The VT procedure produces $\mathcal{M}'=(V',E', \mathbf{b}')$ using the following steps.

\subsubsection{Step 1: Target Vocabulary Alignment}
We align the model with a well-normalized, lowercased target vocabulary $V'$, selected from existing Language Models. This choice allows us to leverage their pre-trained word embeddings for semantic initialization. 

\subsubsection{Step 2: Embedding and Bias Initialization}
Let $O = V \cap V'$ denote overlapping tokens and $N' = V' \setminus O$ new tokens. Simple topological initialization is insufficient as it ignores discrepancies in prior probabilities. We propose a joint initialization of spatial embeddings and biases:

\noindent\textbf{Semantic Initialization via Spatial Topology.}
Our goal is to initialize $E'\in\mathbb{R}^{|V'|\times d}$ so that (i) tokens shared by both vocabularies preserve the pretrained model’s geometry, and (ii) newly introduced units land in semantically plausible regions of the \emph{source} embedding manifold, avoiding random starts that would otherwise require long adaptation. This initialization aims to minimize the representation discrepancy (related to $\varepsilon_{\mathrm{RC}}$ in Theorem~\ref{thm:sc}), ensuring that the starting point of the migration satisfies the compatibility assumptions made in our theoretical framework.
For overlap tokens $u\in O$, we keep the pretrained parameters unchanged:
$E'_u \leftarrow E_u$.
For a new token $t\in N'$, we \emph{transfer neighborhoods} from the target embedding space into the source space using the overlap set $O$ as anchors.

Let $\tilde{E}\in\mathbb{R}^{|V'|\times d}$ denote pretrained embeddings associated with the target vocabulary (e.g., \path{bert-base-uncased}).
We first compute an affinity vector $\mathbf{s}_t\in\mathbb{R}^{|O|}$ between $t$ and each anchor $u\in O$ using cosine similarity:
\begin{equation}
    s_{t,u}=\cos(\tilde{E}_t,\tilde{E}_u)=\frac{\tilde{E}_t^\top \tilde{E}_u}{\|\tilde{E}_t\|\,\|\tilde{E}_u\|}.
\end{equation}
A dense interpolation over all anchors is undesirable: it blurs semantic neighborhoods and may introduce spurious mass on weakly related anchors, especially when $|O|$ is large.
We therefore convert affinities into a \emph{sparse} convex weighting $\boldsymbol{\alpha}_t$ by projecting onto the simplex with \path{sparsemax}~\cite{martins2016softmax}:
\begin{equation}
    \boldsymbol{\alpha}_t=\text{sparsemax}(\mathbf{s}_t)
    =\mathop{\mathrm{argmin}}_{\mathbf{p}\in\Delta^{|O|-1}}\|\mathbf{p}-\mathbf{s}_t\|^2,
\end{equation}
which yields $\boldsymbol{\alpha}_t\ge 0$, $\sum_{u\in O}\alpha_{t,u}=1$, and only a small subset of nonzero neighbors.
Finally, we synthesize the source-space initialization by barycentric interpolation over the corresponding \emph{source} anchors:
\begin{equation}
    E'_t \leftarrow \sum_{u\in O} \alpha_{t,u}\,E_u.
\end{equation}
This constructs $E'_t$ as the point whose local neighborhood (with respect to anchors) matches that of $t$ in the target space, effectively preserving \emph{relative} semantic topology while staying on the pretrained source manifold.
In practice, this produces meaningful embeddings before any MLM adaptation, substantially reducing the optimization burden for newly introduced tokens.

Alternatively, for custom vocabularies lacking a corresponding pre-trained model, we employ \textbf{Sub-token Initialization}. We tokenize each new token $t$ using the source tokenizer into constituent sub-tokens and initialize $E'_t$ as the mean of their source embeddings. This constructs a semantic approximation from the source model's existing sub-word units.

\noindent\textbf{Prior-Aware Distribution Alignment.}
The output bias $\mathbf{b}$ captures unigram priors. To map the target vocabulary's prior structure into the source model's dynamic range, we apply a \emph{Z-score distribution transfer}:
\begin{equation}
\mathbf{b}' \leftarrow \mu(\mathbf{b}_{src}) + \sigma(\mathbf{b}_{src}) \cdot \frac{\mathbf{b}_{tgt} - \mu(\mathbf{b}_{tgt})}{\sigma(\mathbf{b}_{tgt})}
\end{equation}
where $\mu(\cdot)$ and $\sigma(\cdot)$ denote mean and standard deviation. This ensures common words in $V'$ receive higher initial biases while strictly adhering to the logit scale expected by the source Transformer.

\subsubsection{Step 3: Discrepancy-Aware Adaptation}
We freeze the Transformer layers and update only the embedding layer via a short Masked Language Modeling (MLM) phase, utilizing two efficiency mechanisms:

\noindent\textbf{Overlap-Aware Masking Curriculum.} 
Uniform masking is inefficient since tokens in $O$ are already well-learned. We use importance sampling for masking probabilities $P_{\text{mask}}(t) \propto \omega_t$, where $\omega_t = 1$ if $t \in O$ and $\omega_t = \lambda$ if $t \in N'$ ($\lambda > 1$). This curriculum focuses gradients on unaligned regions to accelerate convergence while preventing catastrophic forgetting.

\noindent\textbf{Activation Potential Calibration.} 
Sparse retrievers like SPLADE rely on ReLU to induce sparsity. However, improper initialization hinders adaptation: globally low logits lead to low activation rates and ``dead neurons,'' while globally high logits result in high activation rates and large magnitudes, producing excessive inner products that cause ``dense collapse.'' We shift the bias $\mathbf{b}'$ by a scalar $c$ after MLM ($\mathbf{b}' \leftarrow \mathbf{b}' - c$), where $c$ is determined by probing a subset of the training data. This calibration places the activation rate within a moderate range and ensures that non-zero activations follow a long-tail distribution from zero to the maximum value, creating an ideal regime for margin-based distillation.

\begin{table*}[t]
\centering
\caption{\textbf{Main Results on BEIR.} We report efficiency metrics (document length and FLOPs) and nDCG\textsubscript{10} performance for each dataset. $\dagger$ indicates models provided by us.  The best performance results are \textbf{bolded}.}
\label{tab:beir_detailed_results}
\resizebox{\textwidth}{!}{
\begin{tabular}{l cc | ccccccccccccc | c}
\toprule
 & \multicolumn{2}{c|}{\textbf{Efficiency}} & \multicolumn{13}{c|}{\textbf{Performance per Dataset (nDCG\textsubscript{10})}} & \textbf{Avg} \\
\cmidrule(lr){2-3} \cmidrule(lr){4-16} \cmidrule(l){17-17}
\textbf{Model} & \textbf{Doc\_Len} & \textbf{FLOPs} & \textbf{Arg} & \textbf{Cli} & \textbf{DBP} & \textbf{FEV} & \textbf{FiQA} & \textbf{Hot} & \textbf{NFC} & \textbf{NQ} & \textbf{Quo} & \textbf{SCI} & \textbf{SciF} & \textbf{Tou} & \textbf{TREC} & \textbf{nDCG} \\
\midrule
Co-SelfDistil~\cite{formal2022distillation} & 197.0 & 11.1 & 49.3 & 26.2 & 44.1 & 81.6 & 36.2 & \textbf{69.3} & 35.4 & 54.2 & \textbf{85.0} & \textbf{16.0} & 71.5 & 24.8 & 72.4 & 51.2 \\
Co-EnsembleDistil~\cite{formal2022distillation} & 159.5 & 7.8 & 50.8 & 24.4 & 43.6 & 80.0 & 35.5 & 68.7 & 35.3 & 53.9 & 83.4 & 15.8 & 70.8 & 27.3 & 72.5 & 50.9 \\
SPLADE-v3~\cite{lassance2024splade} & 213.2 & 8.0 & 48.7 & 25.6 & 45.1 & 81.0 & \textbf{38.1} & 69.2 & \textbf{36.3} & \textbf{58.7} & 81.4 & 15.6 & \textbf{71.6} & 31.2 & 73.1 & 52.0 \\
BERT (Splade) $^\dagger$ & 198.7 & 13.0 & \textbf{51.6} & 24.7 & 43.4 & 81.0 & 35.1 & 69.2 & 34.9 & 53.8 & 80.7 & 15.7 & 71.3 & 25.8 & 70.8 & 50.6 \\
\textbf{ModernBERT-VT} $^\dagger$ & 203.1 & 9.1 & 49.0 & \textbf{27.8} & \textbf{45.1} & \textbf{84.4} & 36.3 & 68.8 & 35.8 & 55.9 & 83.6 & 15.6 & 70.5 & \textbf{33.0} & \textbf{75.9} & \textbf{52.4} \\
\bottomrule
\end{tabular}
}
\end{table*}

\section{Experimental Setup}
\label{sec:experimental_setup}

\subsection{Datasets and Evaluation Metrics}

\paragraph{Training Data}
For the \textbf{MLM adaptation} phase, we use the combined English Wikipedia and BookCorpus~\cite{zhu2015aligning} datasets, comprising approximately 6.2 million documents and 3.7 billion tokens.
For the \textbf{sparse retrieval fine-tuning}, we utilize the MS MARCO Passage Retrieval dataset~\cite{msmarco}. Specifically, we use the \textbf{identical training data with \path{SPLADE-EnsembleDistil}}, i.e., the \path{msmarco-hard-negatives}\footnote{https://huggingface.co/datasets/sentence-transformers/msmarco-hard-negatives} dataset, which includes hard negatives mined from MS MARCO and supervisory scores generated by a cross-encoder teacher model\footnote{https://huggingface.co/cross-encoder/ms-marco-MiniLM-L6-v2}.

\paragraph{Evaluation Benchmarks}
For in-domain evaluation, we report \textbf{MRR@10} and \textbf{Recall@1000} on the MS MARCO official development set (Dev).
We also ensure robust evaluation using the TREC-DL 2019~\cite{craswell2025overview} query sets, reporting \textbf{nDCG@10} and \textbf{Recall@1000}.
To evaluate \textbf{zero-shot} generalization, we test on the BEIR benchmark~\cite{thakur2021beir}. Following previous work~\cite{formal2021splade, splade2, lassance2022efficiency, lassance2024splade, geng2024towards}, we use a subset of 13 datasets: \textit{TREC-COVID, NFCorpus, NQ, HotpotQA, FiQA-2018, ArguAna, Webis-Touché2020, DBPedia-Entity, SCIDOCS, FEVER, Climate-FEVER, SciFact, and Quora}.
We report \textbf{nDCG@10}.

\subsection{Models and Baselines}
Our analysis centers on the performance disparity between established encoders and modern architectures in sparse retrieval.

\paragraph{Backbones}
We utilize \path{answerdotai/ModernBERT-base}~\cite{warner2025smarter} as our primary modern backbone.
To demonstrate our VT method, we migrate this model to the vocabulary of \path{bert-base-uncased}. In this transfer, overlapping tokens account for 56.4\% of the vocabulary, while new tokens account for 43.6\%.
We compare these against standard BERT-based sparse implementations.

\paragraph{Baselines}
We compare our models against a comprehensive set of baselines categorized into two groups.
\textbf{Reference Baselines:} These include the traditional lexical model BM25 and established dense and sparse neural retrievers such as DPR, CoCondenser, ColBERTv2, uniCOIL, DeepImpact~\cite{deepimpact}, and DeeperImpact.
Baseline results are sourced directly from their respective original publications.
\textbf{SPLADE Family \& Derivatives:} We compare against the standard SPLADE models, specifically CoCondenser-SelfDistil and CoCondenser-EnsembleDistil~\cite{formal2022distillation}, as well as the recent SPLADE-v3~\cite{lassance2024splade} and ESPLADE~\cite{dudek2023learning}.
To ensure a fair comparison and eliminate discrepancies arising from different evaluation pipelines, we re-evaluated all available open-source SPLADE checkpoints using our own evaluation pipeline.

\subsection{Implementation Details}

\paragraph{Training Protocol}
For the \textbf{VT Adaptation (MLM)} phase, we train for 20k steps (approx. 1 epoch) on Wikipedia and BookCorpus.
We use the AdamW optimizer with a learning rate of 3e-4, a cosine scheduler with 4,000 warmup steps, and a global batch size of 2,048.
Input sequences are truncated to 128 tokens, and the MLM masking probability is set to 0.3.
In the Overlap-Aware Masking Curriculum, the importance sampling weight for new tokens is set to $\lambda=2$.
For the \textbf{Activation Potential Calibration (APC)}, we set the scalar shift $c=5$, which results in an initial activation rate of approximately 40\%.

For \textbf{Sparse Retrieval Fine-tuning}, our training pipeline and configuration align strictly with established protocols of \path{SPLADE-EnsembleDistil}~\cite{formal2022distillation} to isolate the impact of the backbone and vocabulary.
We employ a knowledge distillation approach.
The model is fine-tuned using the MarginMSE loss combined with FLOPs regularization~\cite{paria2020minimizing} to control sparsity.
All models are trained with a maximum sequence length of 256 tokens.

\paragraph{Evaluation Configuration}
During Evaluation, document inputs are truncated to a maximum length of 512 tokens.
We use OpenSearch\footnote{https://opensearch.org/} as our lexical search engine to construct the inverted index and perform the retrieval process.
Metrics are computed using the official BEIR evaluation toolkit.

\paragraph{Reproducibility}
To ensure reproducibility, we fix the random seed to 42 for all experiments.
We execute training for a fixed 150k steps and select the final checkpoint for evaluation.
The implementation are based on PyTorch~\cite{paszke2019pytorch} and HuggingFace transformers~\cite{wolf2019huggingface} library.
All experiments were conducted on 8 NVIDIA A100 Tensor Core GPUs (80GB VRAM).

\section{Results and Analysis}

\begin{table}[t]
\centering
\caption{\textbf{Main Results on BEIR (OOD) and MS MARCO (In-Domain).}
For SPLADE-based models, we report knowledge distillation (KD) teacher numbers.
We report average nDCG\textsubscript{10} on BEIR.
$\dagger$ indicates models provided by us.
For our primary model (ModernBERT-VT), subscripts denote the standard deviation over 5 random seeds; other entries use a single seed (42).
The best results in each group are \textbf{bolded}, and the best overall results are \underline{underlined}.}
\label{tab:main}
\resizebox{\columnwidth}{!}{
\begin{tabular}{l c | c | cc cc}
\toprule
\textbf{Model} & \textbf{Teach.} & \textbf{BEIR} & \multicolumn{2}{c}{\textbf{MSM}} & \multicolumn{2}{c}{\textbf{DL-19}} \\
 &  &  & \textbf{MRR\textsubscript{10}} & \textbf{R\textsubscript{1k}} & \textbf{nDCG\textsubscript{10}} & \textbf{R\textsubscript{1k}} \\
\midrule

\multicolumn{7}{l}{\textit{\textbf{Reference Baselines}}} \\
BM25 
 & -
 & 43.7 
 & 18.4 & 85.3 & 50.6 & 74.5 \\

DPR~\cite{karpukhin2020dense} 
 & -
 & 37.5 
 & 31.9 & 94.1 & 61.1 & 74.2 \\

CoCondenser~\cite{gao2021condenser} 
 & -
 & 42.0 
 & 38.2 & 98.4 & 67.4 & 82.0 \\

ColBERTv2~\cite{santhanam2022colbertv2} 
 & -
 & \textbf{50.0} 
 & \textbf{39.7} & \textbf{98.5} & \underline{\textbf{74.4}} & \underline{\textbf{88.2}} \\

uniCOIL~\cite{lin2021few,thakur2023sprint}
 & -
 & 44.1
 & 35.1 & - & 69.3 & - \\

DeepImpact~\cite{deepimpact,thakur2023sprint}
 & -
 & 41.5
 & 32.7 & 94.8 & 69.5 & - \\

DeeperImpact~\cite{basnet2024deeperimpact}
 & -
 & - 
 & 37.3 & 96.8 & - & - \\

BERT (dense) $^\dagger$
 & -
 & 42.1 
 & 32.9 & 93.3 & 65.7 & 64.7 \\

ModernBERT (dense) $^\dagger$
 & -
 & 44.2 
 & 32.9 & 94.8 & 63.7 & 65.7 \\
\midrule

\multicolumn{7}{l}{\textit{\textbf{SPLADE Family \& Derivates}}} \\
Co-SelfDistil~\cite{formal2022distillation} 
 & 1
 & 51.2 
 & 37.5 & 98.5 & \textbf{73.5} & \textbf{83.4} \\

Co-EnsembleDistil~\cite{formal2022distillation} 
 & 1
 & 50.9 
 & 38.3 & 98.3 & 73.1 & 83.0 \\

ESPLADE~\cite{dudek2023learning}
& 1
& 51.2
& 38.1 & 98.3 & - & - \\

SPLADE-v3~\cite{lassance2024splade} 
 & 5
 & \textbf{52.0} 
 & \underline{\textbf{40.0}} & \underline{\textbf{98.7}} & 72.7 & 83.2 \\

\midrule
\multicolumn{7}{l}{\textit{\textbf{Our Results $^\dagger$}}} \\

BERT (Splade) 
 & 1
 & 50.6
 & 37.5 & 98.2 & 72.8 & 82.1 \\

\hspace{1em} + further MLM 
 & 1
 & 50.5
 & 37.7 & 98.1 & 72.6 & 82.2 \\

ModernBERT (Splade) 
 & 1
 & 47.7 
 & 35.7 & 98.0 & 66.3 & 79.9 \\

\hspace{1em} + further MLM 
 & 1
 & 47.5
 & 35.4 & 97.8 & 69.0 & 79.4 \\

\hspace{1em} + lowercase input 
 & 1
 & 46.8 
 & 36.3 & 97.9 & 69.2 & 80.5 \\

\textbf{ModernBERT-VT}
 & 1
 & \underline{\textbf{52.4}}$_{0.04}$
 & \textbf{38.3}$_{0.08}$ & \textbf{98.4}$_{0.02}$ & \textbf{73.7}$_{0.7}$ & \textbf{83.0}$_{0.3}$ \\

\bottomrule
\end{tabular}
}
\end{table}

\subsection{RQ1: Effectiveness of Vocabulary Transfer}
\label{sec:rq1}

We analyze VT's effectiveness in enabling advanced encoders for sparse retrieval, comparing against baselines from Section~\ref{sec:experimental_setup}. We include SPLADE-v3 as a reference point using publicly released checkpoints; however, SPLADE-v3 employs additional training structure changes and ensemble teacher scores, making it not strictly controlled under our training setting. Therefore, our main comparisons focus on SPLADE-ensemble-distill, which matches our supervision and pipeline.

\subsubsection{The Vocabulary Gap Anomaly}
As shown in Table~\ref{tab:main}, the dense retrieval results confirm the superiority of modern architectures: \textit{ModernBERT (dense)} achieves an average nDCG@10 of 44.2 on BEIR, outperforming \textit{BERT (dense)} at 42.1.
However, this advantage vanishes when applying standard sparse fine-tuning.
Naive \textit{ModernBERT (Splade)} lags behind the older \textit{Co-Ensemble Distil} (47.7 vs 50.9 on BEIR; 35.7 vs 38.3 MRR on MS MARCO), confirming that advanced encoders with raw vocabularies are ill-suited for sparse retrieval without adaptation.

\subsubsection{Effectiveness of Vocabulary Transfer}
Applying VT effectively bridges this gap.
Our \textit{ModernBERT-VT} model achieves a BEIR score of \textbf{52.4}, representing a substantial improvement over the naive implementation.
In our internal comparison of "controlled" implementations (rows marked with $\dagger$), we observe distinct trends.

\noindent\textbf{Naive vs. Lowercase.} Simply forcing the input to be lowercased (\textit{+ lowercase input}) provides mixed results.
While it offers a slight relief for in-domain matching (improving MRR from 35.7 to 36.3), it degrades zero-shot performance on BEIR (dropping from 47.7 to 46.8).
This suggests that while preprocessing can help align some surface forms, it fails to fully leverage the semantic capacity of the backbone and breaks the token integrity required for generalization.

\noindent\textbf{Ineffectiveness of Further MLM.}
Simply continuing MLM on the target corpus (\textit{+ further MLM}) fails to improve performance (BEIR drops to 47.5), confirming the gains stem from VT rather than extra computation.

\noindent\textbf{The VT Advantage.} \textit{ModernBERT-VT} outperforms both the naive and lowercased variants by a wide margin, validating that proper vocabulary adaptation is essential to release the potential of ModernBERT.

\subsubsection{Comparison with State-of-the-Art}

\textit{ModernBERT-VT} establishes a new state-of-the-art on BEIR compared to other sparse retrievers. On BEIR, \textit{ModernBERT-VT} (52.4) outperforms \textit{Co-Ensemble Distil} (50.9) and even \textit{SPLADE-v3} (52.0), despite the latter utilizing a complex 5-teacher ensemble training pipeline.
Table~\ref{tab:beir_detailed_results} confirms these gains maintain efficiency comparable to SPLADE baselines, validating that VT effectively unlocks ModernBERT's robust generalization capabilities.
On MS MARCO, \textit{ModernBERT-VT} (38.3 MRR@10) surpasses all baselines except \textit{SPLADE-v3} (40.0).
Crucially, \textit{SPLADE-v3} benefits from orthogonal enhancements (multi-stage ensemble distillation), whereas we use a standard single-teacher setup to strictly isolate the backbone's impact.
To verify statistical reliability, we additionally train \textit{ModernBERT-VT} with five random seeds: the resulting small standard deviations (BEIR: $\pm 0.04$; MRR@10: $\pm 0.08$) confirm that our reported gains over baselines are not due to seed variance.

\subsubsection{Summary}
Results show the "lag" in advanced encoders is a vocabulary alignment issue, not architectural.
VT unlocks ModernBERT's reasoning power for sparse retrieval, retaining its superior generalization and achieving competitive in-domain performance, effectively closing the dense-sparse gap.

\subsection{RQ2: Impact of Initialization Strategies and VT Components}
\label{sec:rq2}

To verify our VT recipe, we conduct ablation studies on BEIR (OOD) and MS MARCO (In-Domain), examining embedding initialization strategies and adaptation objectives (PDA, OMC, APC).
Results are summarized in Table~\ref{tab:ablation}.

\begin{table}[t]
\centering
\caption{Ablation on BEIR (OOD) and MS MARCO (In-Domain).
We study initialization strategies and VT components.
``Direct'' denotes fine-tuning without MLM; ``Adapted'' applies 20k steps (approx. 1 epoch) MLM before fine-tuning.}
\label{tab:ablation}
\resizebox{0.8\columnwidth}{!}{
\begin{tabular}{l cc cc}
\toprule
& \multicolumn{2}{c}{Direct FT} & \multicolumn{2}{c}{Adapted} \\
\cmidrule(lr){2-3} \cmidrule(lr){4-5}
Model & BEIR & MSM & BEIR & MSM \\
\midrule
\multicolumn{5}{l}{\textit{Baselines \& Init}} \\

ModernBERT (Splade)
 & 47.7 & 35.7 & - & - \\

VT (RandAll-Init)
 & 39.4 & 26.8 & 51.7 & 37.8 \\

 VT (RandNew-Init)
 & 50.2 & 37.0 & 52.0 & 38.1 \\

VT (Mean-Init)
 & 49.7 & 36.9 & 51.8 & 38.2 \\

VT (SubToken-Init)
 & 49.3 & \textbf{37.7} & 52.3 & 37.8 \\

\textbf{VT (Semantic-Init)}
 & \textbf{51.1} & 37.5
 & \textbf{52.4} & \textbf{38.3} \\

 \quad \myarrow 500 steps MLM & - & - & 52.2 & 38.0 \\
\midrule

\multicolumn{5}{l}{\textit{Component Ablation (Semantic-Init)}} \\

\quad w/o PDA
 & 50.9 & 37.5 & 52.2 & 38.1 \\

\quad w/o OMC
 & - & - & 52.1 & 38.0 \\

\quad w/o APC
 & 50.7 & 37.3 & 48.8 & 37.4 \\
\bottomrule
\end{tabular}
}
\end{table}

\subsubsection{Impact of Initialization Strategies}
We compare \textit{Semantic-Init} against several baselines: \textit{Rand-All} (randomizes all embeddings); \textit{Rand-New} (randomizes only non-overlapping tokens); \textit{Mean-Init} (sets new tokens to the overlapping vocabulary centroid); and \textit{SubToken-Init} (averages constituent sub-tokens).
We report results for two settings: ``Direct'' and ``Adapted''.

\noindent\textbf{Semantic-Init is superior to all other methods.}
Table~\ref{tab:ablation} shows that \textit{Semantic-Init} consistently outperforms other strategies.
In the Direct setting, it provides a strong ``warm-start'' (51.1 on BEIR).
MLM adaptation further boosts this to optimal performance (52.4).
Surprisingly, just 500 MLM steps yield near-optimal results (52.2), validating that our method offers a high-quality starting point requiring minimal gradient updates to align the vocabulary.

\noindent\textbf{SubToken-Init is robust but suffers from semantic collapse in Direct FT.}
\textit{SubToken-Init} is a competitive baseline when target embeddings are missing, matching \textit{Semantic-Init} after adaptation (52.3 vs 52.4 on BEIR).
However, in Direct fine-tuning, its high In-Domain accuracy (37.7) contrasts with a significant OOD drop (49.3).
This discrepancy implies ``semantic collapse,'' where embeddings overfit the training domain while drifting from intrinsic semantics.
Table~\ref{tab:case-study} illustrates why \textit{SubToken-Init} can be suboptimal. While effective for clean splits (e.g., ``nationalists''), \textit{SubToken-Init} fails on ambiguous fragments like ``clears'' ($\rightarrow$ ``cle'', ``ars'') or ``centimetres'' ($\rightarrow$ ``cent'', ``imet'', ``res'').
In contrast, \textit{Semantic-Init} identifies robust neighbors independent of surface forms.

\noindent\textbf{Normalization compensates for random initialization.}
Surprisingly, even \textit{Rand-New} surpasses the \textit{ModernBERT (Splade)} baseline in both settings.
This suggests that the gains from normalizing the sparse output space outweigh the degradation caused by randomly initializing a portion of the vocabulary.


\newcommand{\stA}[1]{\textcolor{blue!70!black}{\texttt{#1}}}
\newcommand{\stB}[1]{\textcolor{orange!85!black}{\texttt{#1}}}
\newcommand{\stC}[1]{\textcolor{teal!70!black}{\texttt{#1}}}
\newcommand{\stD}[1]{\textcolor{purple!70!black}{\texttt{#1}}}
\newcommand{\topk}[1]{\texttt{#1}}

\begin{table}[t]
\centering
\caption{Case study of sub-token decompositions and top-5 semantic neighbors and weights for new tokens.}
\small
\setlength{\tabcolsep}{0pt}
\renewcommand{\arraystretch}{1.25}

\begin{tabular}{@{}p{\columnwidth}@{}}
\toprule

\textbf{1) Word:} \stA{collabor}\stB{ations}
\quad \textbf{Sub-tokens:} \texttt{["Ġcollabor", "ations"]} \\
\textbf{Top-5:} (\topk{Ġcollaborate}, 0.123), (\topk{Ġcollaboration}, 0.118), (\topk{Ġcollaborators}, 0.107),
(\topk{Ġcollaborated}, 0.098), (\topk{Ġcollaborative}, 0.090) \\[4pt]

\textbf{2) Word:} \stA{national}\stB{ists}
\quad \textbf{Sub-tokens:} \texttt{["Ġnational", "ists"]} \\
\textbf{Top-5:} (\topk{nationalist}, 0.192), (\topk{nationalism}, 0.126), (\topk{liberals}, 0.056),
(\topk{conservatives}, 0.044), (\topk{militants}, 0.039) \\[4pt]

\textbf{3) Word:} \stA{cle}\stB{ars}
\quad \textbf{Sub-tokens:} \texttt{["Ġcle", "ars"]} \\
\textbf{Top-5:} (\topk{cleared}, 0.165), (\topk{clearing}, 0.064), (\topk{clearer}, 0.049),
(\topk{removes}, 0.043), (\topk{facilitates}, 0.030) \\[4pt]

\textbf{4) Word:} \stA{cent}\stC{imet}\stB{res}
\quad \textbf{Sub-tokens:} \texttt{["Ġcent", "imet", "res"]} \\
\textbf{Top-5:} (\topk{centimeters}, 0.178), (\topk{cm}, 0.169), (\topk{kilograms}, 0.068),
(\topk{inches}, 0.056), (\topk{kilometres}, 0.051) \\

\bottomrule
\end{tabular}
\label{tab:case-study}
\end{table}

\subsubsection{Impact of Adaptation Components}
We dissect the VT adaptation phase by removing specific components from the optimal \textit{Semantic-Init} configuration.

\noindent\textbf{Prior-Aware Distribution Alignment (PDA) and Overlap-Aware Masking Curriculum (OMC).}
Both components are vital for performance. Removing PDA degrades in-domain (MS MARCO MRR: $38.3 \rightarrow 38.1$) and OOD metrics, while removing OMC leads to sub-optimal convergence. Together, they ensure the model learns new vocabulary effectively without overfitting or forgetting pre-trained knowledge of overlapping tokens.

\noindent\textbf{The Critical Role of Activation Potential Calibration (APC).}
Removing APC causes the most significant impact. We find that MLM adaptation sharpens output logits, making them ill-positioned for SPLADE's ReLU activation. This causes instability: without APC, the ``Adapted'' model performs worse on BEIR (48.8) than the ``Direct'' baseline (50.7). APC realigns the activation potential, allowing sparsity regularization to operate effectively and translating semantic gains into retrieval performance.

\begin{table}[t]
\centering
\caption{\textbf{Sensitivity analysis of APC target activation rate} on ModernBERT-VT. We report BEIR (nDCG\textsubscript{10}), FLOPs, and MS MARCO (MRR\textsubscript{10}, R\textsubscript{1k}). Our default $c{=}5$ falls between the 30\% and 40\% activation rate settings.}
\label{tab:apc_sensitivity}
\resizebox{0.7\columnwidth}{!}{
\begin{tabular}{c cc ccc}
\toprule
\textbf{Act. Rate} & \textbf{BEIR} & \textbf{FLOPs} & \textbf{MRR\textsubscript{10}} & \textbf{R\textsubscript{1k}} \\
\midrule
10\% & 51.8 & 10.69 & 38.5 & 98.4 \\
20\% & 52.1 & 10.20 & 38.3 & 98.4 \\
30\% & 52.2 & 9.56 & 38.0 & 98.4 \\
40\% & 52.4 & 8.92 & 38.3 & 98.4 \\
50\% & 52.4 & 7.96 & 38.0 & 98.5 \\
60\% & 52.0 & 6.59 & 37.8 & 98.3 \\
70\% & 51.0 & 5.29 & 37.4 & 98.3 \\
80\% & 48.9 & 4.46 & 37.3 & 98.2 \\
90\% & 48.0 & 4.25 & 37.2 & 98.1 \\
\bottomrule
\end{tabular}
}
\end{table}

\noindent\textbf{Sensitivity to APC Target Activation Rate.}
Table~\ref{tab:apc_sensitivity} reports results across target activation rates from 10\% to 90\%.
BEIR performance forms a plateau in the 30--50\% activation rate range (52.2--52.4), indicating that the method is robust to the exact choice of $c$ and does not require per-backbone tuning.
Above 60\% activation rate, excessive sparsity suppression starves the sparse representation of active dimensions, degrading OOD generalization sharply. Below 20\%, the increased activation density raises FLOPs without improving retrieval quality.
In-domain MS MARCO metrics remain stable throughout (MRR@10: 37.2--38.5), confirming that APC primarily governs the effectiveness--efficiency trade-off for out-of-domain transfer.
Our default $c{=}5$ falls between the 30\% and 40\% activation rate settings, sitting squarely within this plateau and balancing strong BEIR performance with moderate computational cost.

\subsection{RQ3: Full-Vocabulary Transfer vs. Head-Only Adaptation}
\label{sec:rq_vs_esplade}

We compare full-vocabulary transfer (VT) against \textit{head-only} adaptation, a natural alternative where the backbone tokenizer is kept intact while only the output decoder head is replaced to match the target vocabulary. This design, exemplified by ESPLADE~\cite{dudek2023learning}, aims to decouple the encoder’s input space from its output feature space.

\noindent\textbf{Experimental Setup} 
We implement \textbf{ModernBERT-ESPLADE} by replacing only the importance head with our semantic initialization. We further evaluate two adaptation strategies: (1) \textbf{+EMLM}, which applies the unsupervised task from~\cite{dudek2023learning}; and (2) \textbf{+SAP}, which distills sparse lexical representations from a strong teacher (SPLADE-v3) following MILCO~\cite{nguyen2025milco}.

\noindent\textbf{VT Outperforms Head-only Adaptation} 
As shown in Table~\ref{tab:vs_esplade}, full-vocabulary transfer consistently dominates head-only designs. Even without adaptation, VT surpasses ModernBERT-ESPLADE by +1.7 BEIR nDCG@10, suggesting that aligning the input interface with the target domain is crucial. While SAP improves performance, it remains inferior to VT across all benchmarks (e.g., 50.5 vs. 52.4 on BEIR) and requires an additional dependency on a teacher model.

\begin{table}[t]
\centering
\caption{Full-vocabulary transfer vs.\ ESPLADE-style head-only adaptation.}
\label{tab:vs_esplade}
\resizebox{\columnwidth}{!}{
\begin{tabular}{l | c | cc cc}
\toprule
\textbf{Model} & \textbf{BEIR} & \multicolumn{2}{c}{\textbf{MSM}} & \multicolumn{2}{c}{\textbf{DL-19}} \\
 &  & \textbf{MRR\textsubscript{10}} & \textbf{R\textsubscript{1k}} & \textbf{nDCG\textsubscript{10}} & \textbf{R\textsubscript{1k}} \\
\midrule
\multicolumn{6}{l}{\textit{Direct Fine-tuning (No Adaptation)}} \\
ModernBERT-ESPLADE & 49.4 & 36.7 & \textbf{98.4} & \textbf{72.4} & \textbf{80.0} \\
\textbf{ModernBERT-VT (Ours)} & \textbf{51.1} & \textbf{37.5} & 97.9 & 71.9 & 79.7 \\

\midrule
\multicolumn{6}{l}{\textit{Adapted Models}} \\
ModernBERT-ESPLADE + EMLM & 48.0 & 35.6 & 98.3 & 69.1 & 79.6 \\
ModernBERT-ESPLADE + SAP & 50.5 & 37.6 & 98.3 & 71.4 & 81.9 \\
\textbf{ModernBERT-VT (Ours)} & \underline{\textbf{52.4}} & \underline{\textbf{38.3}} & \underline{\textbf{98.4}} & \underline{\textbf{73.7}} & \underline{\textbf{83.0}} \\
\bottomrule
\end{tabular}
}
\end{table}

\noindent\textbf{The Alignment Bottleneck} 
Notably, EMLM degrades performance in our setting (49.4 $\rightarrow$ 48.0). We attribute this to the \textbf{many-to-one mapping conflict}: when the target vocabulary is more granular than the source (target tokens $N >$ source tokens $M$), ESPLADE’s ``first-overlap'' supervision forces $N$ labels onto a single source position. This ill-defined mapping creates ambiguous training signals that confuse the model.
Importantly, this issue is avoided in the \emph{original} EMLM setting, where the target vocabulary is a \emph{word-level unigram} lexicon instead of sub-tokens. 
In contrast, VT unifies the input and output tokenization for MLM training.

\noindent\textbf{Conclusion} 
Modifying only the language-model head is sub-optimal. 
Because the input embedding space remains unchanged, the encoder still organizes representations around the \emph{source} subword units, and the new head must learn a difficult post-hoc translation into the target sparse vocabulary from sparse retrieval supervision alone. 
VT is essential to fully synchronize the model's internal reasoning with the sparse retrieval objective.

\subsection{RQ4: Generalization Across Different Backbones}

To assess the universality of our approach, we extend \textbf{VT} to several widely used encoder architectures, including \path{RoBERTa-base}, \path{RoBERTa-large}, and \path{BERT-base-cased}. Note that RoBERTa models employ a case-sensitive BPE tokenizer with a vocabulary size of 50,265, while \path{BERT-base-cased} uses a case-sensitive WordPiece tokenizer with 28,996 tokens. In this section, we use VT to migrate each backbone’s native vocabulary to the vocabulary of \path{bert-base-uncased}, while keeping all other settings identical to those used in Table~\ref{tab:main}, including the training protocol and distillation pipeline.
Specifically, for RoBERTa, overlapping tokens account for 59.6\% and additional tokens for 40.4\%; for \path{BERT-base-cased}, overlapping tokens constitute 60.9\% and new tokens 39.1\%.
The results in Table~\ref{tab:backbone_transfer} reveal several key insights:

\begin{table}[t]
\centering
\caption{\textbf{Generalization of VT Across Other Backbones.} 
We report average nDCG\textsubscript{10} on BEIR, and MRR/nDCG metrics for MS MARCO and DL-2019. 
The best results for each backbone are \textbf{bolded}.}
\label{tab:backbone_transfer}
\resizebox{0.9\columnwidth}{!}{
\begin{tabular}{l | c | cc cc}
\toprule
\textbf{Model} & \textbf{BEIR} & \multicolumn{2}{c}{\textbf{MSM}} & \multicolumn{2}{c}{\textbf{DL-19}} \\
 &  & \textbf{MRR\textsubscript{10}} & \textbf{R\textsubscript{1k}} & \textbf{nDCG\textsubscript{10}} & \textbf{R\textsubscript{1k}} \\
\midrule

\multicolumn{6}{l}{\textit{\textbf{RoBERTa-base}}} \\
SPLADE 
 & 48.0 
 & 35.1 & 97.8 & 69.4 & 78.6 \\

\hspace{1em} + lowercase input
 & 48.7 
 & 35.7 & 97.9 & 70.5 & 80.0 \\

\textbf{VT (Ours)}
 & \textbf{50.4} 
 & \textbf{37.7} & \textbf{98.3} & \textbf{74.3} & \textbf{82.5} \\
\midrule

\multicolumn{6}{l}{\textit{\textbf{RoBERTa-large}}} \\
SPLADE 
 & 1.4 
 & 0.7 & 19.9 & 0.2 & 4.6 \\

\hspace{1em} + lowercase input
 & 0.8 
 & 0.1 & 8.9 & 0.0 & 0.6 \\

\textbf{VT (Ours)}
 & \textbf{51.3} 
 & \textbf{38.7} & \textbf{98.3} & \textbf{73.0} & \textbf{81.6} \\
\midrule

\multicolumn{6}{l}{\textit{\textbf{bert-base-cased}}} \\
SPLADE 
 & 10.6 
 & 5.9 & 57.5 & 15.2 & 27.7 \\

\hspace{1em} + lowercase input
 & 27.8 
 & 22.8 & 85.4 & 44.8 & 50.6 \\

\textbf{VT (Ours)}
 & \textbf{50.2}
 & \textbf{37.4} & \textbf{98.2} & \textbf{71.7} & \textbf{82.1} \\

\bottomrule
\end{tabular}
}
\end{table}

\noindent\textbf{The vocabulary gap is universal.} The performance degradation is not unique to ModernBERT: standard SPLADE fine-tuning on RoBERTa-base (48.0 nDCG@10 on BEIR) lags behind the BERT-base baseline (50.6). More strikingly, RoBERTa-large and BERT-base-cased suffer from \textit{dense collapse} due to universally high activation values, yielding near-zero performance (e.g., 1.4 and 10.6 on BEIR). This indicates that raw, case-sensitive vocabularies are fundamentally incompatible with sparse retrieval objectives across model families and scales.

\noindent\textbf{VT consistently restores and improves performance.} Applying VT closes these gaps across all backbones. In particular, VT elevates RoBERTa-large from a non-functional regime (1.4) to strong performance (51.3 on BEIR), surpassing the BERT-base baseline and matching the gains observed with ModernBERT.

\noindent\textbf{VT is robust to model scale and tokenizer type.} Its outstanding performance on both RoBERTa-large (355M parameters) and BERT-base-cased (wordpiece tokenizer) highlights VT’s ability to seamlessly adapt to larger parameter spaces and diverse tokenization schemes, all without the need for any additional continuous pre-training.

Overall, these results suggest that VT is a \emph{model-agnostic} solution, providing a robust pathway for integrating diverse encoder architectures into the sparse retrieval paradigm by directly resolving the underlying lexical mismatch.

\subsection{RQ5: Generalization to Inference-Free Retrieval}

To verify the robustness of VT across different sparse retrieval architectures, we evaluate its performance in an inference-free setting~\cite{geng2024towards, shen2025exploring}.
Following the experimental protocol of \citet{shen2025exploring}, we maintain an identical training pipeline, dataset, and hyperparameter configuration, replacing only the backbone encoder to ModernBERT to isolate the impact of VT.
The results in Table~\ref{tab:inference_free_lsr} reveal several key insights:

\begin{table}[t]
\centering
\caption{Performance of VT-adapted models on Inference-free LSR tasks. Follow \citet{shen2025exploring}, we employ IDF enhancement~\cite{geng2024towards} and $\ell_0$ masked flops~\cite{shen2025exploring}.}
\label{tab:inference_free_lsr}
\resizebox{\columnwidth}{!}{
\begin{tabular}{l cc ccc}
\toprule
Model & +IDF & +$\ell_0$ & nDCG\textsubscript{10} & FLOPS & doc\_len \\
\midrule
\multicolumn{6}{l}{\textit{Baselines}} \\
BM25 & -- & -- & 44.5 & -- & -- \\
SPLADE-v3-doc~\cite{lassance2024splade} & -- & -- & 46.8 & 3.3 & 240.3 \\
\midrule
\multicolumn{6}{l}{\textit{Inference-free Models}} \\
\citet{nardini2025effective} & \checkmark & -- & 48.9 & - & - \\
\citet{geng2024towards} & \checkmark & -- & 50.4 & 2.1 & 248.6 \\
CoCondenser-IDF~\cite{shen2025exploring} & \checkmark & & 49.5 & 2.4 & 327.2 \\
CoCondenser-$l_0$~\cite{shen2025exploring} & \checkmark & \checkmark & 50.3 & 2.1 & 275.0 \\
\midrule
\multicolumn{6}{l}{\textit{ModernBERT Adaptation} $\dagger$} \\
ModernBERT (naive) & \checkmark & \checkmark & 44.4 & 1.7 & 337.4 \\
\quad + lowercase input & \checkmark & \checkmark & 49.5 & 2.1 & 296.9 \\
ModernBERT-VT & \checkmark & & 51.4 & 2.1 & 286.8 \\
\textbf{ModernBERT-VT} & \checkmark & \checkmark & \textbf{51.5} & \textbf{2.1} & \textbf{265.3} \\
\bottomrule
\end{tabular}
}
\end{table}

\noindent\textbf{Persistence of the Vocabulary Gap.} Similar to the standard SPLADE setting, the naive ModernBERT underperforms significantly in the inference-free regime, with its 44.4 NDCG@10 failing to even match the BM25 baseline.
Notably, the naive model yields longer document lengths yet lower FLOPS than baselines, suggesting that the lexical mismatch of surface variants is more pronounced in the inference-free architecture.

\noindent\textbf{Effectiveness of VT.} ModernBERT-VT effectively bridges this gap, achieving a state-of-the-art NDCG@10 of 51.5.
This represents a substantial improvement over both the SPLADE-v3 baseline and the $\ell_0$-enhanced models \cite{shen2025exploring}.

\noindent\textbf{Efficiency Synergy.} When combined with $l_0$-flops and $l_0$-activation, ModernBERT-VT achieves the best balance between effectiveness and efficiency, yielding the shortest average document length (265.3) and competitive FLOPs (2.1).

These findings demonstrate that the benefits of VT are not limited to dual-encoder sparse retrievers but extend to high-efficiency, inference-free indices, allowing modern backbones to realize their full potential in latency-critical applications.

\subsection{RQ6: Domain Specialization via Vocabulary Transfer}
\label{sec:rq_domain}

Learned Sparse Retrieval (LSR) is particularly sensitive to tokenization in specialized domains (e.g., Chemistry), where general-purpose tokenizers often over-fragment technical terms. We investigate whether VT can effectively migrate a general backbone to a \textit{domain-specific} vocabulary synthesized from scratch.

\noindent\textbf{Experimental Setup}
We train five BPE tokenizers (sizes 10k--50k) on the \path{dolma-chem} corpus~\cite{basfai_dolma_chem_only_query_generated_2025}, utilizing BERT-style normalization (lowercase, stripping accents). We migrate ModernBERT-base to these vocabularies using VT with sub-token initialization, followed by a brief MLM adaptation (3k steps). Models are fine-tuned on 200k chemistry query-document pairs using InfoNCE loss. We evaluate on \textbf{ChemHotpotQA} and \textbf{ChemNQ}~\cite{kasmaee2024chemteb}. 
Due to the limited size of the training data, we observed variance in model performance. To ensure the reliability of our results, we report the mean and standard deviation across five random seeds.

\begin{table}[t]
\centering
\caption{\textbf{Effectiveness of Domain-Specific Vocabulary Transfer.} 
Results are nDCG\textsubscript{10} (Mean $\pm$ Std) over 5 runs. ``Frag.'' denotes the fragmentation rate (Tokens/Word).}
\label{tab:domain_results}
\resizebox{\columnwidth}{!}{
\begin{tabular}{l r | cc | cc}
\toprule
\textbf{Model} & \textbf{Vocab} & \multicolumn{2}{c|}{\textbf{ChemHotpotQA}} & \multicolumn{2}{c}{\textbf{ChemNQ}} \\
    & \textbf{Size} & \textbf{Frag.} & \textbf{nDCG\textsubscript{10}} & \textbf{Frag.} & \textbf{nDCG\textsubscript{10}} \\
\midrule
ModernBERT-base & 50k & 1.36 & $0.526 \pm 0.081$ & 1.33 & $0.271 \pm 0.031$ \\
\midrule
ModernBERT-VT & 10k & 1.53 & $0.690 \pm 0.010$ & 1.49 & $\mathbf{0.412 \pm 0.016}$ \\
ModernBERT-VT & 20k & 1.40 & $0.720 \pm 0.039$ & 1.38 & $0.342 \pm 0.023$ \\
ModernBERT-VT & 30k & 1.35 & $0.688 \pm 0.025$ & 1.33 & $0.329 \pm 0.011$ \\
ModernBERT-VT & 40k & 1.32 & $\mathbf{0.756 \pm 0.033}$ & 1.31 & $0.322 \pm 0.045$ \\
ModernBERT-VT & 50k & 1.30 & $0.737 \pm 0.026$ & 1.30 & $0.307 \pm 0.021$ \\
\bottomrule
\end{tabular}
}
\end{table}

\noindent\textbf{Results and analysis}
VT markedly improves cross-domain transfer to chemistry, yielding large gains over the original
general-domain vocabulary on both datasets.
Across vocab sizes, fragmentation decreases monotonically as expected, but retrieval performance does
not: ChemHotpotQA peaks at 40k, whereas ChemNQ peaks at 10k.
This decoupling indicates that reduced token fragmentation is beneficial but insufficient to predict
adaptation quality.
With limited domain MLM (3k steps) and only 200k supervised pairs, larger vocabularies introduce more
rare sub-tokens whose embeddings and lexical weights are weakly trained, which can increase variance
and hurt generalization (notably on \textsc{ChemNQ}).
Overall, these results support VT as an effective mechanism for \emph{cross-domain} adaptation of sparse
retrievers, while highlighting that vocab-size selection remains a non-trivial trade-off between
tokenization adequacy and parameter/data efficiency.

\section{Conclusion}
In this work, we demonstrate that the performance degradation of advanced encoders in sparse retrieval is not an architectural deficiency but a consequence of the vocabulary gap. By shifting from lossless-reconstruction-oriented modern tokenization to normalized, representation-compatible vocabularies, we unlock the reasoning power of next-generation backbones for lexical matching. Our proposed VT method provides a robust and efficient solution, allowing models like ModernBERT and RoBERTa to adapt to sparse-friendly vocabularies via geometric initialization and minimal adaptation steps. VT not only restores the competitiveness of these models but establishes new state-of-the-art results across out-of-domain benchmarks and inference-free architectures. Ultimately, our findings highlight that vocabulary design is a fundamental bottleneck in neural sparse retrieval and offer a universal recipe to bridge the divide between modern foundation models and sparse retrieval objectives.

\bibliographystyle{ACM-Reference-Format}
\bibliography{sample-base}

\appendix

\section{Per-Seed Results for ModernBERT-VT}
\label{app:per-seed}

\begin{table}[h]
\centering
\caption{\textbf{Per-seed results for ModernBERT-VT.} We report BEIR nDCG\textsubscript{10}, MS MARCO MRR\textsubscript{10} / R\textsubscript{1k}, and TREC DL-19 nDCG\textsubscript{10} / R\textsubscript{1k}. The last two rows summarize mean and standard deviation over the five seeds.}
\label{tab:per_seed}
\resizebox{0.9\columnwidth}{!}{
\begin{tabular}{l c cc cc}
\toprule
\textbf{Seed} & \textbf{BEIR} & \multicolumn{2}{c}{\textbf{MS MARCO}} & \multicolumn{2}{c}{\textbf{DL-19}} \\
 & nDCG\textsubscript{10} & MRR\textsubscript{10} & R\textsubscript{1k} & nDCG\textsubscript{10} & R\textsubscript{1k} \\
\midrule
42 & 52.44 & 38.38 & 98.37 & 73.90 & 83.10 \\
1 & 52.36 & 38.27 & 98.40 & 72.93 & 83.30 \\
2 & 52.38 & 38.17 & 98.41 & 73.15 & 83.02 \\
3 & 52.36 & 38.32 & 98.39 & 74.73 & 82.66 \\
4 & 52.35 & 38.23 & 98.40 & 73.98 & 82.77 \\
\midrule
Mean & 52.38 & 38.27 & 98.39 & 73.74 & 82.97 \\
Std  & 0.04  & 0.08  & 0.02  & 0.69  & 0.27 \\
\bottomrule
\end{tabular}
}
\end{table}

To support the mean $\pm$ std statistics reported for ModernBERT-VT in Table~\ref{tab:main}, we list the raw per-seed results in Table~\ref{tab:per_seed}. All five runs share identical training configurations, differing only in the random seed used for data shuffling and parameter initialization of non-transferred components. The seed originally reported in the main text is \textbf{42}; the additional four seeds (\texttt{1}, \texttt{2}, \texttt{3}, \texttt{4}) were run post-hoc to quantify variance.

\section{Proof of Theorem~\ref{thm:sc}}
\label{app:proof-thm-sc}

For a score class $\mathcal{F}$, write the empirical Rademacher complexity
$\hat{\mathfrak{R}}_n(\mathcal{F})
=\mathbb{E}_{\sigma}\!\big[\sup_{f\in\mathcal{F}}
\frac{1}{n}\sum_{i=1}^n \sigma_i f(z_i)\big]$.

\begin{lemma}[Uniform generalization for bounded $L$-Lipschitz losses]
\label{lem:UG}
Let $\ell:\mathbb{R}\to[0,1]$ be $L$-Lipschitz and $\mathcal{F}$ be any real-valued score class.
Then for any $\delta\in(0,1)$, with probability at least $1-\delta$,
\[
\sup_{f\in\mathcal{F}}\Big(\mathcal{L}(f)-\hat{\mathcal{L}}_n(f)\Big)
\ \le\ 2L\,\hat{\mathfrak{R}}_n(\mathcal{F})\ +\ C_{\mathrm{gen}}\sqrt{\tfrac{\log(1/\delta)}{n}},
\]
for a universal constant $C_{\mathrm{gen}}>0$.
\end{lemma}

\begin{proof}
Let $\mathcal{G}=\{\ell\circ f:f\in\mathcal{F}\}\subset[0,1]$.
Standard symmetrization gives
$\mathbb{E}\big[\sup_{g\in\mathcal{G}}(\mathbb{P}g-\hat{\mathbb{P}}g)\big]\le 2\hat{\mathfrak{R}}_n(\mathcal{G})$,
and Ledoux--Talagrand contraction yields
$\hat{\mathfrak{R}}_n(\mathcal{G})\le L\,\hat{\mathfrak{R}}_n(\mathcal{F})$
(e.g.,~\cite{ledoux1991probability, Mohri2018}).
Since $\mathcal{G}\subset[0,1]$, changing one sample changes the supremum by at most $1/n$,
so McDiarmid's inequality converts the expectation bound to the stated high-probability form,
absorbing numerical constants into $C_{\mathrm{gen}}$~\cite{mcdiarmid1989method}.
\end{proof}

\begin{lemma}[RC transfer between $V$ and $V'$]\label{lem:rc-transfer}
Under \eqref{eq:RC} and $L$-Lipschitz $\ell$,
\[
\inf_{h\in\mathcal{H}_{V'}}\mathcal{L}(h)
\ \le\
\inf_{h\in\mathcal{H}_V}\mathcal{L}(h)
\ +\ L\,\varepsilon_{\mathrm{RC}}.
\]
\end{lemma}

\begin{proof}
Fix $\beta\ge0$ with $\|\beta\|_1\le B$. By \eqref{eq:RC}, there exists $\beta'\ge0$, $\|\beta'\|_1\le B$ such that
$\sup_{(q,d)}|\langle\beta,u_\theta(q,d)\rangle-\langle\beta',u'_\theta(q,d)\rangle|\le\varepsilon_{\mathrm{RC}}$.
Applying $L$-Lipschitz $\ell$, taking expectations, and then taking infima over feasible $\beta$ proves the claim.
\end{proof}

\paragraph{Proof of Theorem~\ref{thm:sc}.}
Let $\hat h_G\in\mathcal{H}_{V'}$ be an ERM and let $h^\star\in\arg\min_{h\in\mathcal{H}_{V'}}\mathcal{L}(h)$.
By ERM optimality,
\begin{align}
\mathcal{L}(\hat h_G)-\mathcal{L}(h^\star)
\;&\le\;
\big(\mathcal{L}(\hat h_G)-\hat{\mathcal{L}}_n(\hat h_G)\big)
+\big(\hat{\mathcal{L}}_n(h^\star)-\mathcal{L}(h^\star)\big) \nonumber\\
&\le\;
2\sup_{h\in\mathcal{H}_{V'}}
\big(\mathcal{L}(h)-\hat{\mathcal{L}}_n(h)\big).
\end{align}

Applying Lemma~\ref{lem:UG} to $\mathcal{H}_{V'}$ and using \eqref{eq:rad-composition} gives, with prob.\ $\ge 1-\delta$,
\begin{align}
\mathcal{L}(\hat h_G)-\inf_{h\in\mathcal{H}_{V'}}\mathcal{L}(h)
&\le
4L\,\hat{\mathfrak{R}}_n(\mathcal{H}_{V'})
+2C_{\mathrm{gen}}\sqrt{\tfrac{\log(1/\delta)}{n}} \nonumber\\
&\le
4BL\,\hat{\mathfrak{R}}_n(\mathcal{W}_{V'};\|\cdot\|_\infty)
+2C_{\mathrm{gen}}\sqrt{\tfrac{\log(1/\delta)}{n}}.
\label{eq:est-core}
\end{align}
Finally, Lemma~\ref{lem:rc-transfer} yields
$\inf_{h\in\mathcal{H}_{V'}}\mathcal{L}(h)\le \inf_{h\in\mathcal{H}_{V}}\mathcal{L}(h)+L\varepsilon_{\mathrm{RC}}$,
and combining with \eqref{eq:est-core} proves \eqref{eq:main-bound}.\qed

\end{document}